\documentclass[final,11pt,letterpaper]{elsarticle}
\usepackage[top=1in, bottom=1.2in, left=1.15in, right=1.1in]{geometry}
\usepackage{amsmath}
\usepackage{amsthm}
\usepackage{amssymb}
\usepackage{bm}  % Must be after txfonts
\usepackage{url}
\usepackage{enumerate}
\usepackage{graphicx}
\usepackage{blkarray}
\usepackage{color}
\usepackage{caption}
\usepackage{subcaption}
\usepackage{sidecap}
\usepackage{wrapfig}
\usepackage{flushend}
\usepackage{lineno}
\usepackage[]{algorithm2e}
\usepackage{hyperref}
\newcommand{\Oh}{\mathcal{O}}

\newcommand{\conv}{\mathop{\mathrm{conv}}}
\newcommand{\vertexset}{\mathop{\mathrm{vert}}}
\newcommand{\xc}{\mathop{\mathrm{xc}}} % extension complexity

\newcommand{\np}{{\mathrm{NP}}}
\newcommand{\ptime}{{\mathrm{P}}}

\newcommand{\RR}{{\mathbb{R}}}

\newcommand{\CutP}{\mathrm{CUT}^\square}

\newcommand{\SAT}{{\mathrm{SAT}}}
\newcommand{\UNSAT}{{\mathrm{UNSAT}}}

\newcommand{\cf}{\mathcal{CF}}
\newcommand{\cnf}{\mathrm{CNF}}
\newcommand{\cutsat}{\mathrm{CUTSAT}}
\newcommand{\xor}{\oplus}
\newcommand{\nspace}[2]{\ensuremath{#1\text{-}\mathrm{NSPACE}(#2)}}
\newcommand{\dspace}[2]{\ensuremath{#1\text{-}\mathrm{DSPACE}(#2)}}
\newcommand{\dlin}{\ensuremath{\mathrm{DLIN}}}
\newcommand{\walk}{\omega}
\newcommand{\signature}{\sigma}
\newcommand{\pwalk}[2]{\mathrm{P_{markov}}({#2},#1)}
\def\uch{\uplus} %convex hull of union
\def\hypercube{\bm{\square}}
\def\vzero{\bm{0}}
\def\ee{\bm{e}}
\newcommand{\DEF}[1]{\slshape{#1}\normalfont}
\def\ff{\bm{f}}
\def\xx{\bm{x}}
\def\yy{\bm{y}}
\def\zz{\bm{z}}
\def\uu{\bm{u}}
\def\LL{\bm{L}}
\def\oneL{\mathrm{1L}}
\def\oneNL{\mathrm{1NL}}
\newcommand{\logspace}{\mathrm{LOGSPACE}}

\newcommand{\NL}{\mathrm{NL}}

\theoremstyle{plain}
\newtheorem{theorem}{Theorem}
\newtheorem*{theorem*}{Theorem}
\newtheorem{corollary}{Corollary}
\newtheorem{lemma}{Lemma}
\newtheorem{proposition}{Proposition}

\theoremstyle{definition}
\newtheorem{definition}{Definition}
\newtheorem{example}{Example}
\theoremstyle{remark}

\journal{arXiv}
\begin{document}

\begin{titlepage}

\begin{frontmatter}

\title{Extension Complexity of Formal Languages\footnote{This work was partially supported by grant no. GA15-11559S of GA{\v{C}R}}}
%\author{Hans Raj Tiwary\corref{cor1}}
%\ead{hansraj@kam.mff.cuni.cz}

%\address{KAM/ITI,
 %Charles University,
 %Malostransk\'e n\'am. 25,
 %118 00 Prague 1, Czech Republic}

%\cortext[cor1]{KAM/ITI,
% Charles University,
% Malostransk\'e n\'am. 25,
% 118 00 Prague 1, Czech Republic}

\author{Hans Raj Tiwary\\ 
	\texttt{hansraj@kam.mff.cuni.cz}\\
	KAM/ITI, Charles University,\\
	Malostransk\'e n\'am. 25,\\
	118 00 Prague 1, Czech Republic}

\begin{abstract} 
In this article we undertake a study of extension complexity from the perspective of formal languages. We define a natural way to associate a family of polytopes with binary languages. This allows us to define the notion of extension complexity of formal languages. We prove several closure properties of languages admitting compact extended formulations. Furthermore, we give a sufficient machine characterization of compact languages. We demonstrate the utility of this machine characterization by obtaining upper bounds for polytopes for problems in nondeterministic logspace; lower bounds in streaming models; and upper bounds on extension complexities of several polytopes. %The upper bounds obtained by our method match best known bounds while our lower bounds in the streaming model allow for nondeterministic algorithms. Our results demonstrate how many different results can be unified.
\end{abstract}

\begin{keyword}
Extended formulations \sep formal languages
\end{keyword}
\end{frontmatter}
\end{titlepage}

%\linenumbers
%\newpage
%%%%%%%%%%%%% languages-core.tex
\section{Introduction}
A polytope $Q$ is said to be an extended formulation of a polytope $P$ if $P$ can be described as a projection of $Q$. Measuring the size of a polytope by the number of inequalities required to describe it, one can define the notion of \emph{extension complexity} of a polytope $P$ -- denoted by $\xc(P)$ -- to be the size of the smallest possible extended formulation.

Let $\bm{\varphi}$ be a boolean formula. Consider the following polytopes:

$$\begin{array}{lcl}
\mathrm{SAT}&=&\conv\left\{\xx\left|\xx\text{ encodes a satisfiable boolean
formula}\right.\right\},\\
\mathrm{SAT}(\bm{\varphi})&=&\conv\left\{\xx\left|\bm{\varphi}(\xx)=1\right.\right\}.
\end{array}$$

%\marginpar{check placement of footnote}
The former polytope consists of all strings that encode\footnote{Assume
some (arbitrary but fixed) encoding of boolean formulae as binary strings.}
satisfiable boolean formulae, while the latter language consists of all
satisfying assignments of a given boolean formula. Which of these represents the
boolean satisfiability problem \emph{more naturally}?

Reasonable people will agree that there is no correct choice of a natural
polytope for a problem. One complication is that there are various kinds of
problems: decision, optimization, enumeration, etc, and very similar problems
can have very different behavior if the notion of problem changes.

Several recent results have established superpolynomial lower bounds on the extension complexity of specific polytopes. For example Fiorini et al. \cite{jacm/FioriniMPTW15} showed that polytopes associated with MAX-CUT, TSP, and Independent Set problems do not admit polynomial sized extended formulations. Shortly afterward Avis and the present author \cite{mp/AvisT15} showed that the same holds for polytopes related to many other NP-hard problems. Subsequently Rothvo{\ss} \cite{stoc/Rothvoss14} showed that even the perfect matching polytope does not admit polynomial sized extended formulation. These results have been generalized in multiple directions and various lower bounds have been proved related to approximation \cite{BJLP2013, CLRS13, bfps2012} and semidefinite extensions \cite{FP13, mp/BDP15, LT12, stoc/LeeRS15}.

\medskip
A few fundamental questions may be raised about such results.
\begin{itemize}
	\item How does one choose (a family of) polytopes for a specific problem?
	\item To what extent does this choice affect the relation between extension complexity of the chosen polytope and the complexity of the underlying problem?
	\item What good are extension complexity bounds anyway?\footnote{Perfect Matching remains an easy problem despite exponential lower bound on the extension complexity of the perfect matching polytope. What does an exponential lower bound for the cut polytope tell us about the difficulty of the MAX-CUT problem?}
\end{itemize}

The intent of this article is to say something useful (and hopefully interesting) about such problems. In particular, our main contributions are the following.
\begin{itemize}
	\item We define formally the notion of extension complexity of binary language. Our definition is fairly natural and we do not claim any novelty here. This however is a required step towards any systematic study of problems that admit small extended formulations.
	\item We define formally what it means to say that a language admits small extended formulation. Again we do not claim novelty here since Rothvo{\ss} mentions similar notion in one of the first articles showing the existence of polytopes with high extension complexity \cite{mp/Rothvoss13}.
	\item We prove several closure properties of languages that admit compact extended formulations. Some of these results are trivial and some follows from existing results. For a small number of them we need to provide new arguments.
	\item We prove a sufficient condition in terms of walks on graphs and in terms of accepting Turing Machines, for a language to have polynomial extension complexity. We show how this characterization can be used to prove space lower bounds for non-deterministic streaming algorithms, and also to construct compact extended formulations for various problems by means of a small ``verifier algorithm''. We provide some small examples to this end.
\end{itemize}

\section{Background Material and Related Work}
\subsection{Polytopes and Extended Formulations}
A polytope $P\subseteq \mathbb{R}^d$ is a closed convex set defined as intersection of a finite number of inequalities. Alternatively, it can be defined as the convex hull of a finite number of points. Any polytope that is full-dimensional has a unique representation in terms of the smallest number of defining inequalities or points. The \emph{size} of a polytope is defined to be the smallest number of inequalities required to define it. For the purposes of this article all polytopes will be assumed to be full-dimensional. While in doing so, no generality is lost for our discussion, we will refrain from discussing such finer points.  We refer the reader to \cite{Ziegler/95} for background on polytopes.

A polytope $Q$ is called an \emph{Extended Formulation} or simply \emph{EF} of a polytope $P$, if $P$ can be obtained as a projection of $Q$. The \emph{extension complexity} of a polytope, denoted by $\xc(P)$, is defined to be the smallest size of any possible EF of $P$.

Extended formulations have a long history of study. We mentioned several of them in the introduction. Here we refer to only a handful of work that are closely related to this article. For more complete picture related to extended formulations, we refer the reader to the excellent surveys by Conforti et al. \cite{anor/ConfortiCZ13} and by Kaibel \cite{Kaibel11} as a point to start.

We will use the following known results related to extension complexity.

\begin{theorem}[Balas \cite{dam/Balas98}]\label{thm:xc_union}
Let $P_1$ and $P_2$ be polytopes and let $P=P_1\uch P_2$, where $\uch$ denotes the convex hull of the union. Then $\xc(P)\leqslant \xc(P_1)+\xc(P_2)+2$.
\end{theorem}

\subsection{Online Turing machines}
In this article we would be interested in \emph{online} variants of Turing machines. Informally speaking, these machines have access to two tapes: an input tape where the head can only move from left to right (or stay put where it is) and a work tape where the work head can move freely. When the machine halts, the final state determines whether the input has been accepted or not. Such machines - like usual Turing machines - can be either deterministic or non-deterministic. For a non-deterministic machine accepting a binary language $\LL$ we require that if $\xx\notin \LL$ then the machine rejects $\xx$ for all possible non-deterministic choices, and if $\xx\in \LL$ then there is some set of non-deterministic choices that make the machine accept $\LL$.

The working of an online Turing machine can be thought of as the working of an online algorithm that makes a single pass over the input and decides whether to accept or reject the input. Natural extensions allow the machine to make more than one pass over the input. 

\begin{definition}
The complexity class $\nspace{k}{s(n)}$ is the class of languages
accepted by a $k$-pass non-deterministic Turing machines using space $s(n)$. Similarly, the complexity class $\dspace{k}{s(n)}$ is the class of languages accepted by a $k$-pass deterministic Turing machine using space $s(n)$.
\end{definition}

The classes $\oneL$ and $\oneNL$ were introduced by Hartmanis, Mahaney, and Immerman \cite{HarImmMah78, HarMah81} to study weaker forms of reduction. In our terminology the class $\oneL$ would be $\dspace{1}{\log{n}}$ while the class $\oneNL$ would be $\nspace{1}{\log{n}}$. The motivation for defining these classes was that if we do not know whether $\ptime$ is different from $\np$ or not, then using a polynomial reduction may not be completely justified in saying that a problem is as hard or harder than another problem, and weaker reductions are probably more meaningful. In any case, these languages have a rich history of study. It is known that non-determinism makes one-pass machines strictly more powerful for $s(n)=\Omega(\log{n})$ \cite{Szepietowski98}.

\subsection{Glued Product of Polytopes}
Let $P_1\subseteq\mathbb{R}^{d_1+k}$ and $P_2\subseteq\mathbb{R}^{d_2+k}$ be two $0/1$ polytopes with vertices $\vertexset(P_1), \vertexset(P_2)$ respectively. The \emph{glued product} of $P_1$ and $P_2$ where the gluing is done over the last $k$ coordinates is defined to be: $$P_1\times_k P_2:=\conv\left\{\left.\begin{pmatrix}\xx\\\yy\\\zz\end{pmatrix}\in\{0,1\}^{d_1+d_2+k}\right|\begin{pmatrix}\xx\\\zz\end{pmatrix}\in\vertexset(P_1), \begin{pmatrix}\yy\\\zz\end{pmatrix}\in\vertexset(P_2)\right\}.$$

We will use the following known result about glued products.

\begin{lemma}\label{lem:glued_product_xc}\cite{Margot_thesis, CP12}
Let $P_1\subseteq\mathbb{R}^{d_1+k}$ and $P_2\subseteq\mathbb{R}^{d_2+k}$ be two $0/1$ polytopes such that the every vertex of $P_1$ and $P_2$ contains at most one nonzero coordinate entry among the $k$-coordinates used for the gluing. Then,
$$\xc(P_1\times_k P_2)\leqslant \xc(P_1)+\xc(P_2).$$
\end{lemma}

\section{Polytopes for Formal Languages}
Let $\LL\subseteq\{0,1\}^*$ be a language over the $0/1$ alphabet. 
For every natural number $n$ define the set $\LL(n):=\left\{\xx\in\{0,1\}^n\mid
\xx\in\LL\right\}$. Viewing each string $\xx\in \LL(n)$ as a column vector, and
ordering the strings lexicographically, we can view the set $\LL(n)$ as a matrix
of size $n\times|\LL(n)|$. Thus we are in a position to naturally associate a
family of polytopes with a given language and the extension complexity of these
polytopes can serve as a natural measure of how hard is it to model these
languages as Linear Programs.%\marginpar{Perhaps some words of caution here}

That is, one can associate with $\LL,$ the family of polytopes
$\mathcal{P}(\LL)=\{P(\LL(1)),P(\LL(2)),\ldots\}$ with $P(\LL(n)):=\conv\{\xx~|~\xx\in\LL(n)\}$. The 
extension complexity $\xc(\mathcal{P}(\LL))$ can then be defined as a function such that $\xc(\mathcal{P}(\LL)) (n)=\xc(P(\LL(n))).$
The
extension complexity $\xc(\mathcal{P}(\LL))$, as a function of $n$, is then an intrinsic measure of complexity
of the language $\LL$.

\subsection*{Extension complexity of Languages}
\begin{definition}
The \emph{extension complexity} of a language $\LL$ -- denoted by $\xc(\LL)$ --
is defined by $\xc(\LL):=\xc(\mathcal{P}(\LL))$.
\end{definition}

We say that the \emph{extension complexity} of $\LL$, denoted by $\xc(\LL)$ is $\ff(n)$, where $
\ff:\mathbb{N}\to\mathbb{R}_+$ is a non-negative function on 
natural numbers, if for every polytope  $P(\LL(n))\in \mathcal{P}(\LL)$ we have 
that  $\xc(P(\LL(n))) = \ff(n).$ One can immediately see that this definition is 
rather useless in its present form since for different values of $n$, the 
corresponding polytopes in $\mathcal{P}(\LL)$ may have extension complexities 
that are not well described by a simple function. For example, the perfect matching polytope
 would have no strings of length $n$ if $n$ is 
not of the form $\binom{r}{2}$ for some even positive integer $r$. To avoid 
such trivially pathological problems, we will use asymptotic notation to
describe 
the membership extension complexity of languages. 

We will say that $\xc(\LL)=\Oh(\ff)$ to mean that there exists a constant $c>0$ 
and a natural number $n_0$ such that for every polytope $P(\LL(n))\in\mathcal{P}(\LL)$ 
with $n\geqslant n_0$ we have $\xc(P(\LL(n))) \leqslant c\ff(n).$

We will say that $\xc(\LL)=\Omega(\ff)$ to mean that there exists a constant $c>0$ 
such that for every natural number $n_0$ there exists an $n\geqslant n_ 0 $ 
such that $\xc(P(\LL(n))) \geqslant c\ff(n).$ Note the slight difference 
from the usual $\Omega$ notation used in asymptotic analysis of algorithms.\footnote{This usage, however, is common among number theorists.} The intent 
here is to be able to say that a polytope family of a certain language  
contains an infinite family of polytopes that have high extension complexity. 

Finally, we will say that $\xc(\LL)=\Theta(\ff)$ to mean that $\xc(\LL)=\Oh(\ff)$ as 
well as $\xc(\LL)=\Omega(\ff).$ To give an example of the notation, the recent 
result of Rothvo\ss ~\cite{stoc/Rothvoss14} proving that perfect matching polytope has high 
extension complexity would translate in our setting to the following statement.

\begin{theorem*} {\textbf{\cite{stoc/Rothvoss14}}} 
Let $\LL$ be the language consisting of the characteristic vectors of perfect 
matchings of complete graphs. Then, there exists a constant $c>1$ such that $
\xc(\LL)=\Omega(c^n).$
\end{theorem*}

One can extend the above notation to provide more information by being able 
to use functions described by asymptotic notation as well. For example, 
knowing that the perfect matching polytope for $K_n$ has extension complexity 
at most ${{2}}^\frac{n}{2}$ \cite{mp/FaenzaFGT15} together with Rothvo\ss' result one could 
say that the language of all perfect matchings of complete graphs has 
extension complexity $2^{\Theta(n)}.$ 

\begin{proposition} 
For every language $\LL\subseteq\{0,1\}^*$ we have $\xc(\LL) \leqslant \xc(\LL)+\xc(\overline{\LL})\leqslant 2^{n}.$
\end{proposition}
\begin{proof}
The first inequality is trivial. For the last inequality, observe that $\LL(n)$ and $\overline{\LL}(n)$ has at most $2^{n}$ strings altogether which become the vertices of the polytopes.
\end{proof}

\section{Languages with small extension complexities}

Now we are ready to define the class of languages that we are interested in: 
namely, the languages that have small extension complexities. 

\begin{definition}
$\cf$ is the class of languages admitting $\mathcal{C}$ompact extended
$\mathcal{F}$ormulations and is defined as  
\begin{eqnarray*}
\mathcal{CF}&=&\left\{\LL\subseteq \{0,1\}^* ~|~ \exists c>0 \text{ s.t. }
\xc(\LL)\leqslant
n^c\right\}.
\end{eqnarray*}
\end{definition}

For example, regular languages form a subset of the class $\cf$. That is, if $\LL$ is a regular language, then $\LL\in\cf$ \cite{mp/FioriniP15}.

\subsection{Some canonical examples}
For any given boolean formula $\bm{\varphi}$ with $n$ variables define the polytope $\SAT(\bm{\varphi})$ to be the convex hull of all satisfying assignments and $\UNSAT(\bm{\varphi})$ to be the convex hull of all non-satisfying assignments. That is, 
\begin{eqnarray*}
\SAT(\bm{\varphi})&:= &\conv(\{\xx\in\{0,1\}^n\mid \bm{\varphi}(\xx)=1 \}),\\ 
\UNSAT(\bm{\varphi})&:= &\conv(\{\xx\in\{0,1\}^n\mid \bm{\varphi}(\xx)=0 \}).
\end{eqnarray*}
%\subsection{A SAT instance for CUT}

Let $n\in\mathbb{N}$ and $m=n^2.$ For the complete graph $K_n$ define a 3SAT boolean formula $\bm{\varphi}_m$ such that $\CutP(K_n)$ -- the convex hull of all edge-cuts of the complete graph $K_n$ -- is a projection of $\SAT(\bm{\varphi}_m)$ as follows. Consider the relation $\xx_{ij}=\xx_{ii}\xor \xx_{jj},$ where $\xor$ is the xor operator. The boolean formula $$(\xx_{ii}\lor\overline{\xx}_{jj}\lor \xx_{ij})\land(\overline{\xx}_{ii}\lor{\xx_{jj}}\lor {\xx_{ij}})\land({\xx_{ii}}\lor{\xx_{jj}}\lor \overline{\xx}_{ij})\land(\overline{\xx}_{ii}\lor\overline{\xx}_{jj}\lor \overline{\xx}_{ij})$$ is true if and only if $\xx_{ij}=\xx_{ii}\xor \xx_{jj}$ for any assignment of the variables $\xx_{ii},\xx_{jj}$ and $\xx_{ij}.$

Therefore we define $\bm{\varphi}_m$ (with $m=n^2$) as %\marginpar{Formatting?}
\begin{eqnarray}\label{eqn:cnf4cut}\bm{\varphi}_m := \bigwedge_{i,j \in [n] \atop i \neq j} \left[ \begin{array}{l}(\xx_{ii}\lor\overline{\xx}_{jj}\lor \xx_{ij})\land(\overline{\xx}_{ii}\lor{\xx_{jj}}\lor {\xx_{ij}})\land\\({\xx_{ii}}\lor{\xx_{jj}}\lor \overline{\xx}_{ij})\land(\overline{\xx}_{ii}\lor\overline{\xx}_{jj}\lor \overline{\xx}_{ij})\end{array} \right].
\end{eqnarray}

We will call the family of $\cnf$ formulae defined by (\ref{eqn:cnf4cut}) to be the $\cutsat$ family. It is easy to see the following. 

\begin{lemma}\label{lem:xc_sat}
$\xc(\SAT(\bm{\varphi}_m))=2^{\Omega(n)},$ where $m=n^2.$
\end{lemma}
\begin{proof}
The satisfying assignments of $\bm{\varphi}_m$ when restricted to the variables $\xx_{ij}$ with $i\neq j$ are exactly the cut vectors of $K_n$ and every cut vector of $K_n$ can be extended to a satisfying assignment of $\bm{\varphi}$. Therefore $\CutP(K_n)$ is a projection of $\SAT(\bm{\varphi}_m)$. The result then follows from the fact that the extension complexity of the cut polytope $\CutP(K_n)$ is $2^{\Omega(n)}$ \cite{jacm/FioriniMPTW15}.
\end{proof}

\begin{lemma}\label{lem:xc_unsat}
$\xc(\UNSAT(\bm{\varphi}_m))\leqslant \Oh(n^4).$
\end{lemma}
\begin{proof}
Let $\bm{\varphi}$ be a DNF formula with $n$ variables and $m$ clauses/terms. We can show that $\xc(\SAT(\bm{\varphi}))\leqslant \Oh(mn).$

If $\bm{\varphi}$ consists of a single clause then it is just a conjunction of some literals. In this case $\SAT(\bm{\varphi})$ is a face of the $n$-hypercube and has $\xc(\SAT(\bm{\varphi}))\leqslant 2n.$ Furthermore, for DNF formulae $\bm{\varphi}_1,\bm{\varphi}_2$ we have that $\SAT(\bm{\varphi}_1\lor\bm{\varphi}_2)=\SAT(\bm{\varphi}_1)\uch\SAT(\bm{\varphi}_2)$. Therefore, using Theorem \ref{thm:xc_union} repeatedly we obtain that for a DNF formula $\bm{\varphi}$ with $n$ variables and $m$ clauses/terms $\SAT(\bm{\varphi})\leqslant \Oh(mn).$
\end{proof}

\section{Closure properties of compact languages}\label{sec:closure}
{%\marginpar{more operations!}
Now we discuss the closure properties of the class $\cf$ with respect to some 
common operations on formal languages. The operations that we consider are as follows.

$\begin{array}{lcl}
\bullet~\text{\textbf{Complement}} & \text{\textbf{:}} & \overline{\LL}=\{\xx
\mid \xx\notin \LL\} \\
\bullet~\text{\textbf{Union}} & \text{\textbf{:}} & \LL_1\cup \LL_2=\{\xx\mid
\xx\in \LL_1 \vee \xx\in \LL_2\}\\
\bullet~\text{\textbf{Intersection}} & \text{\textbf{:}} & \LL_1\cap
\LL_2=\{\xx\mid \xx\in \LL_1 \wedge \xx\in \LL_2\}\\
\bullet~\text{\textbf{Set difference}} & \text{\textbf{:}} & \LL_1\setminus
\LL_2=\{\xx\mid \xx\in \LL_1 \wedge \xx\notin \LL_2\}\\
\bullet~\text{\textbf{Concatenation}} & \text{\textbf{:}} &
\LL_1\LL_2=\{\xx\yy\mid \xx\in \LL_1 \wedge \yy\in \LL_2\}\\
\bullet~\text{\textbf{Kleene star}} & \text{\textbf{:}} & \LL^{*}=\{\varepsilon\}\cup\LL\cup
\LL\LL\cup \LL\LL\LL \cup \LL\LL\LL\LL \cup \ldots
\end{array}$

\begin{theorem}\label{thm:closure_complement}
$\cf$ is not closed under complement.
\end{theorem}
\begin{proof}
Let $\bm{\Phi}$ be the family of 3CNF formulae containing $\cutsat$ formulae for $m=n^2$ and containing some tautologies with $m$ variables for all other $m$. Let $\LL_{\mathrm{sat}}$ be the language containing the satisfying assignments of the formulae in this family. Similarly, let $\LL_{\mathrm{unsat}}$ be the language containing the non-satisfying assignments of the formulae in this family.

It is easy to see that $\LL_{\mathrm{sat}}=\overline{\LL}_{\mathrm{unsat}}$. Now, $\LL_{\mathrm{unsat}}\in\cf$ due to Lemma \ref{lem:xc_unsat} while $\LL_{\mathrm{sat}}\notin\cf$ due to Lemma \ref{lem:xc_sat}.
\end{proof}

\begin{theorem}\label{thm:closure_union}
$\cf$ is closed under union.
\end{theorem}
\begin{proof}
Let $\LL_1$ and $\LL_2$ be two languages. Then, $\xc(\LL_1\cup\LL_2)\leqslant
\xc(\LL_1)+\xc(\LL_2)+2$ (cf. Theorem \ref{thm:xc_union}).
\end{proof}

\begin{theorem}\label{thm:closure_intersection}
$\cf$ is not closed under intersection.
\end{theorem}
\begin{proof}%\marginpar{Check proof}
Let $\LL_1$ be a language such that a string $\xx\in\LL_1$ if and only if it
satisfies the following properties:
\begin{itemize}
	\item $|\xx|=(n+1)\binom{n}{2}$ for some natural number $n$, where $|\xx|$ is the number of characters in the string $\xx$, and
	\item $\xx_{ij(n+1)}=\xx_{iji}\xor\xx_{ijj}$ if the characters are
indexed as $\xx_{ijk}$ with $1\leqslant i<j\leqslant n$, $1\leqslant k\leqslant
n{+}1.$
\end{itemize}

We claim that $\xc(\LL_1)=\Oh(n)$. Indeed
$P\left(\LL_1\left(\left(n+1\right)\cdot\binom{n}{2}\right)\right)$ is the
product of polytopes
$$P_{ij}=\left\{\xx\in\{0,1\}^{n+1}~|~\xx_{n+1}=\xx_i\xor\xx_j\right\}$$ for
$1\leqslant i<j\leqslant n$ and $\xc(P_{ij})=\Oh(n)$. Therefore, $\xc(P(\LL_1(\Theta(n^3))))=\Oh(n^3)$ and
$\xc(\LL_1)=\Oh(n)$. To see that the extension
complexity of $P_{ij}$ is linear in $n$, note that $P_{ij}$ can be constructed 
as the convex hull of the union of four polytopes $P_{ij}^{ab}$, $a,b\in\{0,1\}$ defined as:
$$P_{ij}^{ab}=\left\{\xx\in\{0,1\}^{n+1}~|~\xx_i=a,\xx_j=b,\xx_{n+1}=a\xor b\right\}.$$
Each $P_{ij}^{ab}$ is isomorphic to a face of the $n$-cube and thus has extension complexity $\Oh(n)$.

Now let $\LL_2$ be a language such that a string $\xx\in\LL_2$ if and only if it
satisfies the following properties.
\begin{itemize}
	\item $|\xx|=(n+1)\binom{n}{2}$ for some natural number $n$, and
	\item $\xx_{i_1j_1k}=\xx_{i_2j_2k}$ for all $k\in[n], i< j\in[n].$
\end{itemize}

Each polytope $P\left(\LL_2\left(\left(n+1\right)\cdot{\binom{n}{2}}
\right)\right)$ is just an embedding of $\hypercube_{n+\binom{n}{2}}$ in
$\RR^{(n+1)\binom{n}{2}}$ where $\hypercube_k$ is the $k$-dimensional hypercube. 
Therefore, $\xc(P(\LL_2(\Theta(n^3))))=\Oh(n^2)$ and so $\xc(\LL_2)=\Oh(n^{2/3}).$

Finally, observe that for $m=(n+1)\binom{n}{2}$ the polytope
$P((\LL_1\cap\LL_2)(m))$ when projected to the coordinates labeled
$\xx_{ij(n+1)}$ is just the polytope $\CutP_n$ (cf. Lemma
\ref{lem:xc_sat}). Therefore, $\xc((\LL_1\cap\LL_2)(\Theta(n^3)))=2^{\Omega(n)}$ and even
though $\LL_1, \LL_2\in\cf$, the intersection $\LL_1\cap\LL_2\notin\cf.$ 
\end{proof}

\begin{theorem}\label{thm:closure_setminus}
$\cf$ is not closed under set difference.
\end{theorem}
\begin{proof}
The complete language $\{0,1\}^*$ clearly belongs to $\cf$. For any language
$\LL$ we have $\overline{\LL}=\{0,1\}^*\setminus \LL$. If $\cf$ were closed
under taking set-difference, it would also be closed under taking complements.
But as pointed out in Theorem \ref{thm:closure_complement}, it is not.
\end{proof}

\begin{theorem}\label{thm:closure_concat}
$\cf$ is closed under concatenation.
\end{theorem}
\begin{proof}
$P(\LL_1\LL_2(n))$ is the union of the polytopes $P(\LL_1(i))\times
P(\LL_2(n-i))$ for $i\in[n]$. Therefore, we have that $\xc(\LL_1\LL_2)\leqslant
\Oh(n(\xc(\LL_1)+\xc(\LL_2))).$
\end{proof}

\begin{theorem}\label{thm:closure_kleene_star}
$\cf$ is closed under Kleene star.
\end{theorem}
\begin{proof}
%\marginpar{$e_i^j$ notation; perhaps an earlier description of $\CP(\LL^*)$?}
Let $\ee^k_l$ denote the $l$-th unit vector of length $k$, that is, a vector of length $k$
all whose entries except the $l$-th one is zero and the $l$-th entry is one. Also,
let $\vzero^k$ denote a vector of all zeroes of length $k$. Now, let $\LL\in\cf$. For $0\leqslant k \leqslant n$, consider the polytope $P_k$ 
defined as
$$P_k:=\conv\left(\left\{\left.\begin{pmatrix}\ee^{n+1}_{i+1}\\\vzero^{i}
\\\xx\\\vzero^{n{-}i{-}k}\\\ee^{n+1}_{i+k+1}
\end{pmatrix}\in\left\{0,1\right\}^{3n+2}\right| \begin{array}{rl}& \xx\in \LL
\\\land&|\xx|=
k\\\land& 0\leqslant i\leqslant n-k \end{array} \right\}\right).$$ 
In particular $P_0$ consists of the vectors $\ee^{n+1}_{i+1},\vzero^n,\ee^{n+1}_{i+1}$ 
for $0\leqslant i\leqslant n$. Define
$P:=\conv(\bigcup_{j=0}^n
P_j).$ Then, $\displaystyle\xc(P) \leqslant \sum_{k=0}^n\xc(P_k) \leqslant
\sum_{k=0}^n
(n\xc(P(\LL(k)))) \leqslant \Oh(n^2\xc(\LL)).$

%\marginpar{Perhaps one needs to discuss earlier if $\xc(\LL)$ is always assumed
%to be increasing? Here $\xc(\LL)=\poly$ so no problem. For binary polytopes as
%long as number of vertices increase the dimension will increase.}
The vertices of the polytope $P$ are $0/1$ vectors of length $3n+2$ with the following
structure. The middle $n$ positions (starting at position $n+2$) contain some $\xx\in\{\varepsilon\}\cup\LL$
padded with some zeroes on the left and the right. We will refer to these as the padded part.
The first $n+1$ coordinates contain 
a one at the $(i+1)$-th position exactly if $i$ zeroes are padded to the left of $\xx$ 
and the last $n+1$ coordinates contain a one at the $(j+1)$-th position exactly if the number
of padded zeroes on the left of $\xx$ together with the length of $\xx$ equals $j$. This structure allows
us to simulate the concatenation of strings in $\LL$ by taking the glued product of $P$ with itself. Take, 
for example, the glued product of $P$ with itself over the last $n{+1}$ coordinates of $P$ and the 
first $n{+}1$ coordinates of $P$. A vertex of this polytope is of the form 
$\ee^{n+1}_{i+1},\vzero^{i},\xx,\vzero^{n{-}i{-}|\xx|},\ee^{n+1}_{i{+}|\xx|{+}1},\vzero^{i{+}|\xx|},\yy,\vzero^{n{-}i{-}|\xx|{-}|\yy|},\ee^{n+1}_{i+|\xx|+|\yy|+1}$ where $\xx,\yy\in\{\varepsilon\}\cup\LL$.
If we add the $n$ coordinates containing the string $\xx$ padded with zeroes to the $n$ coordinates 
containing the string $\yy$ padded with zeroes to make the total length (including the padding) $n$, 
we get the concatenated string $\xx\yy$ padded with 
zeroes. We will use this phenomenon to repeatedly ``concatenate'' strings of $\LL$ by taking
the glued product of $P$ with itself and select the face which corresponds to strings of length $n$ (with no padded zeroes).

Let $S_0$ be the face of $P$ defined by
the first coordinate being $1$ and the following $n$ coordinate being $0$. That is,
if the first $n+1$ coordinates of $P$ are labeled $z_0,\ldots,z_n$ then $S_0$ is 
the face $z_0=1,z_1=0,\ldots,z_n=0$.
This corresponds to selecting strings of $\LL$ with no zeroes padded to the left.
That is, the padded part of vertices of $S_0$ consists of vectors of the form
$\xx,\vzero^{n-|\xx|}$ where $\xx\in\{\varepsilon\}\cup\LL$.
Construct $S_{i{+}1}$ by taking the glued product of $S_i$ with $P$ over the
last
$n{+}1$ coordinates of $S_{i}$ and the first $n{+}1$ coordinates of $P$. If
we only look at the padded parts of the components of $S_k$ we see a 
vector of length $n(k+1)$ consisting of $k+1$ consecutive parts of $n$ consecutive coordinates,
each of which comes from the padded part of one of the component $P$ of $S_k$. 
Any such vector has the form $\xx_0,\vzero^{n-|\xx_0|},\vzero^{|\xx_0|},\xx_1,\vzero^{n-|\xx_0|-|\xx_1|},\ldots,\vzero^{|\xx_0|+\cdots+|\xx_{k-1}|},\xx_k,\vzero^{n-|\xx_0|-\cdots-|\xx_k|}$ where
$\xx_0,\ldots,\xx_k\in\{\varepsilon\}\cup\LL$. Take each of these $k+1$ parts of $n$ coordinates and add them together. We get a vector of the form $\xx_0,\xx_1,\ldots,\xx_k,\vzero^{n-|\xx_0|-\cdots-|\xx_k|}$. 
This corresponds to the concatenated string $\xx_0,\ldots,\xx_k$ padded with zeroes to get a total length of $n$. In particular, if we 
consider the polytope $S_n$ then among its vertices we have all those that correspond to strings of length $n$ in $\LL^*$. These are exactly the vertices we wish to isolate.

Take the face $R$ of $S_{n}$ defined by the penultimate $n$ coordinates being $0$ and
the last coordinate being $1$. That is, if the last $n+1$ coordinates of $S_n$ are labeled
$z_0,\ldots,z_n$ then $R$ is the face $z_0=0,\ldots,z_{n-1}=0,z_n=1$. 
If we only look at the padded parts of vertices in this face, it contains vertices of the form
$\xx_0,\vzero^{n-|\xx_0|},\vzero^{|\xx_0|},\xx_1,\vzero^{n-|\xx_0|-|\xx_1|},\ldots,\vzero^{|\xx_0|+\cdots+|\xx_{n-1}|},\xx_n$ where 
each $\xx_i\in\{\varepsilon\}\cup\LL$. Also, we have that 
$|\xx_0|+\cdots+|\xx_n|=n$. If we add together each consecutive $n$ coordinates containing the padded $\xx_i$ then we 
get the vector $\xx_0,\ldots,\xx_n$ which belongs to $\LL^*$ and has length $n$. Moreover any string in $\LL^*$
of length $n$ can be constructed by concatenating $n+1$ strings $\xx_0,\ldots,\xx_n$ where each $\xx_i\in\{\varepsilon\}\cup\LL$. 
This concatenation would also be simulated in the polytope $S_n$
and would appear in the face $R$. Therefore, $R$ is an EF for
$P(\LL^*(n)).$ 
Moreover, $\xc(R)\leq\xc(S_{n})\leqslant
(n+1)\xc(P) \leqslant  \Oh(n^3\xc(\LL)).$ 

Therefore, $\xc(\LL^*)=\Oh(n^3\xc(\LL))$ and $\LL^*\in\cf.$
\end{proof}

\section{Computational power of compact languages}
We would like to start the discussion in this section by pointing out that the class of compact languages is in some sense too powerful. This power comes just from non-uniformity in the definition.

\begin{proposition}
$\cf$ contains undecidable languages.
\end{proposition}

It is easy to construct undecidable languages that are in $\cf$. Take any uncomputable function $\ff:\mathbb{N}\to\{0,1\}$ and define the language $\LL$ containing all strings of length $n$ if $\ff(n)=1$ and no strings of length $n$ if $\ff(n)=0$. That is, $P(\LL(n))$ is the empty polytope when $\ff(n)=0$ and the $n$-cube when $\ff(n)=1$. So the extension complexity of $\LL$ is $\Theta(n).$

At the moment we do not want to start a discussion about controlling the beast that non-uniformity unleashes. Rather we will focus on something more positive. We will show that if a language is accepted by a non-deterministic $\logspace$ online Turing machine, then its extension complexity is polynomial. This brings into fold many non-regular languages already. And as we will see, this characterization allows us to give simple proofs for polynomial extension complexity for some polytopes.

Before we proceed, we would also like to point out that, in the following discussion, the assumption on the input tape being accessed in a one-way fashion is not something one can remove easily. There are languages in $\logspace$ and $\mathrm{AC}^0$ that have exponential extension complexity: for example, the string of all perfect matchings of $K_n$. As remarked earlier, this language has exponential extension complexity \cite{stoc/Rothvoss14}. A string $\xx$ of this language of length $\binom{n}{2}$ encodes a perfect matching by setting $\xx_{ij}$ to one if the edge $e_{ij}$ is selected. We also have that $\xx_{ij}=\xx_{ji}$ so that the graph is undirected. To see that this language is in $\logspace$ note that a very simple algorithm recognizes whether a given string encodes a perfect matching of $K_n$. First we can loop through all values of $i,j$ to verify that $\xx_{ij}=\xx_{ji}$. Then, for every vertex $i$ and for every edge $e_{ij}$ incident to this vertex we count how many values of $x_{ij}$ are set 
to one and verify that only one edge is selected. All of this can be done with a constant number of pointers. To see that this language is in $\mathrm{AC}^0$ note that we can write the perfect matching predicate as a CNF formula. We only need to encode $\xx_{ij}=\xx_{ji}$ for all $i,j$ and that for a fixed $i$ if $\xx_{ij}=1$ then $\xx_{ik}=0$ for all $k\neq j$. Now for any CNF formula there is an obvious $\mathrm{AC}^0$ circuit of depth two by first turning all disjunctions into OR gates and then the conjunctions as one AND gate.

\subsection{Polytopes of walks in graphs}
 \begin{definition}
 Let $D=(V,A)$ be a directed graph with every edge labeled either zero or one. 
 Consider two nodes $u,v\in V$ and a walk $\walk$ of length $n$ from $u$ to $v$.
 The \DEF{signature} of $\walk$ -- denoted by $\signature_\walk$ -- is the sequence 
 of edge labels along the walk $\walk$. The node $u$ is called the \DEF{source} of the walk 
and the node $v$ the \DEF{destination}.
 \end{definition}
 
 %\begin{example}
%\todo{  A small example of signature.}
 %\end{example}

\begin{definition}
 Consider the convex hull of all zero-one vectors of the form $(u,\sigma,v)$ where $u$ and $v$ are indices of two nodes in $D$ and $\sigma$ is the signature of some 
 walk of length $n$ from $u$ to $v$. This polytope -- denoted by $\pwalk{n}{D}$ -- is called 
 the \DEF{Markovian polytope of $D$}.
\end{definition}

\begin{lemma}\label{lem:xc_walksignatures}
Let $D=(V,A)$ be a directed graph (possibly with self-loops and multiple edges) 
with every edge labeled either zero or one. Then, $\pwalk{n}{D}$ has extension 
complexity at most $2|V|+|A|\cdot n.$
\end{lemma}
\begin{proof}
Let us encode every vertex of $D$ with a zero-one vector of length $|V|$ such that 
the unit vector $\ee_i$ represents vertex $i$.

Define polytope $P_\text{trans}\subset \{0,1\}^{|V|+1+|V|}$ with $(a,z,b) 
\in\{0,1\}^{|V|+1+|V|}$ a vertex of $P_\text{trans}$ if and only if it encodes 
a possible transition in $D$. That is, $a$ and $b$ encode vertices of $V$, and the 
coordinate $z$ represents the label of the edge following which one can move from $a$ to $b$. 
Since $P_\text{trans}$ has at most $|A|$ vertices $\xc(P_ \text{trans}) \leqslant |A|.$ 
%(cf: Proposition \ref{prop:xc_nov}).

Let $P_0$ be the convex hull of $(i,\ee_i)$ for $i\in V$ and $P_f$ be the convex hull
of $(\ee_i,i)$ for $i\in V$. Observe that the two polytopes are the same except
for relabeling of coordinates. Also, $\xc(P_0)=\xc(P_f)\leqslant |V|$.

Let $P_1=P_\text{trans}$. For $2\leqslant i\leqslant n$, 
construct the polytope $P_{i}$ by gluing the last $|V|$ coordinates of $P_{i-1}$ 
with the first $|V|$ coordinates of $P_\text{trans}$. By Lemma \ref{lem:glued_product_xc} we have 
that $\xc(P_{n})\leqslant |E|\cdot n.$ Also, notice that the vertices of $P_n$ exactly encode walks of length $n$ between two vertices of $D$. Every vertex of $P_n$ is of the form $\ee_{i_0},z_1,\ee_{i_1},\ldots,z_n,\ee_{i_n}$ where $i_0,i_1,\ldots,i_n$ is a walk of length $n$ between $i_0$ and $i_n$. The signature of this walk is exactly $z_1,\ldots,z_n$. Furthermore, for any walk $i_0,\ldots,i_n$ of length $n$ between $i_0$ and $i_n$ and with the signature $z_1,\ldots,z_n$, the vector $\ee_{i_0},z_1,\ee_{i_1},\ldots,z_n,\ee_{i_n}$ is a vertex of $P_n$.

Finally, let $P$ be the polytope obtained by gluing last $|V|$ coordinates of $P_0$
with the first $|V|$ coordinates of $P_n$, and then gluing the last $|V|$ vertices of 
the result with the first $|V|$ coordinates of $P_f$. Note that $\xc(P)\leqslant 2|V|+|A|\cdot n.$

To complete the proof, notice that $P$ is an extended formulation for $\pwalk{n}{D}.$ 
In particular, each vertex of $P$ is of the form $i_0,\ee_{i_0},z_1,\ee_{i_1},\ldots,z_n,\ee_{i_n},i_n$, where
$i_0,i_1,\ldots,i_n$ is a walk of length $n$ from $i_0$ to $i_n$. This walk has the signature $z_1,\ldots,z_n$.
Also, for every walk $i_0,\ldots,i_n$ of length $n$ between $i_0$ and $i_n$ with signature $z_1,\ldots,z_n$ we have
$i_0,\ee_{i_0},z_1,\ee_{i_1},\ldots,z_n,\ee_{i_n},i_n$ as a vertex of $P$. Therefore projecting out the coordinates containing values $\ee_{i_j}$ produces 
exactly the vertices of $\pwalk{n}{D}.$ 
\end{proof}

\subsection{Polytopes for Online Turing Machines}
\begin{lemma}\label{lem:kpasstoonepass}
Let $\LL\in\nspace{k}{s(n)}$. Then, $\LL\in\nspace{1}{\Oh(ks(n))}$.
\end{lemma}
\begin{proof}
Let $M_n$ be a $k$-pass nondeterministic Turing machine that accepts strings of length $
n$ and uses space $s(n)$. We will simulate $M_n$ using a multi-tape single pass nondeterministic Turing machine 
called the simulator $S$. $S$ is supplied with $k$ work tapes. $S$ starts 
by guessing the initial work state of $M_n$ at the start of $i$-th pass and 
writing them on the $i$-th work tape. $S$ then simulates (using extra space 
on each work tape) each of the passes independently starting from their 
respective initial configuration. Once the entire input has been scanned, the 
simulator verifies that the work space of $M_n$ on the $i$-th tape at the end of the pass 
matches the guess for the initial content for the $(i+1)$-th tape. $S$ will accept 
only if the last tape is in an accepting state.

To store the content of work tape and the current state, $S$ needs $s(n)+o(s(n))$ space for each pass. Thus $S$ uses a single pass and total space of $ks(n)(1+o(1))$. %By Lemma \ref{lem:simulation_lemma_ntm} the extension complexity of the strings 
%accepted by $M_n$ is then $2^{\Oh(p(n)s(n))}n$.
\end{proof}

Thus for our purposes it suffices to restrict our attention to single pass TMs.

\begin{definition}
The \DEF{configuration graph} for an input of length $n$ for a given one-pass Turing machine (deterministic or non-deterministic) is constructed as follows. For each fixed $n$,
consider the directed graph whose nodes are marked with a label consisting of 
$s(n)+\lceil\log{(s(n))}\rceil$ characters. The labels encode the complete configuration 
of the Turing machine: the content of the work tape and head position on the work tape. We make directed edges between two nodes $u$ 
and $v$ if the machine can reach from configuration $u$ to configuration $v$ by a sequence of transitions with exactly one input bit read in between. The directed edge is labeled by the input bit read during this sequence of transition.

Finally, we add two special nodes: a start node with a directed edge to each possible starting configuration of the machine, and a finish node with a directed edge from each possible accepting configuration. Each of these directed edges are labeled by zero.
\end{definition}

\begin{lemma}\label{lem:cg_size}
The configuration graph for an input of length $n$ for a one-pass Turing machine has $\Oh(2^{s(n)}s(n))$ nodes. If the Turing machine is non-deterministic, this graph has $\Oh(4^{s(n)}(s(n))^2)$ edges. If the Turing machine is deterministic then this graph has $O(2^{s(n)}s(n))$ edges. 
\end{lemma}
\begin{proof}
The bound for number of nodes is clear from the construction of the configuration graph. We can have at most two transition edges between any two (possibly non-distinct) nodes: one corresponding to reading a zero on the input tape, and one corresponding to reading a one. Therefore, asymptotically the configuration graph can have the number of edges that is at most the square of the number of nodes.

For deterministic Turing machines, each node in the configuration graph has exactly two outgoing edges (possibly to the same node). Therefore the number of edges is asymptotically the same as the number of vertices.
\end{proof}

Now Lemma \ref{lem:xc_walksignatures} can be used to bound the extension complexity of languages accepted by one-pass machines.

\begin{theorem}\label{thm:simulation_lemma_ntm} Let $\LL\in\nspace{1}{s(n)}.$
 Then, $\xc(\LL)=\Oh(4^{s(n)}(s(n))^2\cdot n).$ 
\end{theorem} 
\begin{proof}
Let $\LL\in\nspace{1}{s(n)}$ be a 
language. That is, there exists a Turing machine that when supplied with a 
string on the one-way input tape uses at most $s(n)$ cells on the work tape, makes
a single pass over the input and then accepts or rejects the input. If the input 
string is in $\LL,$ some sequence of non-deterministic choices lead the machine to
an accepting state, otherwise the machine always rejects.

The length-$n$ strings that are accepted by such a Turing machine correspond exactly to the signatures of length $n+2$ walks on the corresponding configuration graph $D$. The first and the last character of these strings is always zero. Therefore, an extended formulation for $P(\LL(n))$ is obtained by taking the face of $\pwalk{n+2}{D}$ corresponding to walks that start at the start node and finish at the finish node. By Lemma \ref{lem:xc_walksignatures}, $\pwalk{n+2}{D}$ has extension complexity  $\Oh(4^{s(n)}(s(n))^2\cdot n)$, and so does the desired face.
\end{proof}

If $\LL$ is accepted by a one-pass deterministic TM then one can do better because the configuration graph has fewer edges.

\begin{theorem}\label{thm:simulation_lemma_dtm} Let $\LL\in\dspace{1}{s(n)}.$
 Then, $\xc(\LL)=\Oh(2^{s(n)}s(n)\cdot n).$
\end{theorem} 

\subsection{Extensions for multiple-pass machines}
\begin{theorem}\label{thm:online_ntm} Let $\LL\in\nspace{p(n)}{s(n)}$. Then,  
$\xc(\LL)= 2^{\Oh(p(n)s(n))}n.$
\end{theorem}
\begin{proof}
This follows immediately from Lemma \ref{lem:kpasstoonepass} and Theorem \ref{thm:simulation_lemma_ntm}.
\end{proof}

\begin{theorem} Let $\mathcal{M}$ be a (not necessarily uniform) family of deterministic 
online Turing machines. Let the number of passes and the space 
used by the family be bounded by functions, $p(n),s(n)$
respectively. Let $\LL(\mathcal{M})$ be the language accepted by $\mathcal{M}$. 
Then,  $\xc(\LL(\mathcal{M}))\leqslant 2^{\Oh(p(n)s(n))}n.$
\end{theorem}

\begin{corollary} 
If $\LL$ is accepted by a fixed-pass non-deterministic $\logspace$ Turing machine 
then $\LL\in\cf$.
\end{corollary}

We end this section with the following remark. For a language to be compact (that is, to have polynomial extension complexity), it is sufficient to be accepted by an online Turing machine (deterministic or not) that requires only logarithmic space. However, this requirement is clearly not necessary. This can be proved by contradiction: Suppose that the condition is necessary. Then the class of compact languages must be closed under taking intersection. (Simply chain the two accepting machines and accept only if both do). Since we have already established (cf. Theorem \ref{thm:closure_intersection}) that the class of compact languages is not closed under taking intersection, we have a contradiction.

\section{Applications}\label{sec:xc_1tm_app}
\subsection{Polytopes of certificates: The nondeterministic $\logspace$ class}
Traditionally, the polytope associated with a given problem instance is the convex hull of certificates for that instance. For example, the CUT polytope of a graph is the convex hull of all edge-cuts, the perfect matching polytope of a graph is the convex hull of all perfect matchings, etc. This motivates the following definition of natural polytopes associated with problems.%In such a setting a problem $\Phi$ is defined by a set $I$ of instances and a verifier Turning machine $M_\Phi$. For each instance $\xx\in I$ a set of certificates $C_\Phi(\xx)$ --typically of polynomial length -- is implicitely defined by the Turing machine $M_\Phi$. The ``natural polytope'' for an instance $\xx$ -- denoted by $P_\Phi(\xx)$ -- is the convex hull of certificates in $C_\Phi(\xx)$.

\begin{definition}\label{def:nl_polytopes}
 Let $\LL\subseteq\{0,1\}^*$ be a language and let $M$ be a verifier for certificates for $\LL$. For any instance $\xx\in\{0,1\}^n$ the $\LL_M$-polytope of $\xx$ -- denoted by $P_{(\LL,M)}(\xx)$ -- is defined to be the convex hull of all strings $\yy\in\{0,1\}^{q(n)}$ such that $M(\xx,\yy)=1$ where $M(\xx,\yy)$ denotes the output of $M$ when provided with $\xx$ and $\yy$ on two input tapes.
\end{definition}

Certificate based definition of the class $\np$ is well known, where the certificates are required to be checkable in polynomial time. A less well known certificate based definition is that of the class $\NL$: the class of languages accepted by nondeterministic logspace Turing machines.

\begin{definition}[\cite{Arora_book}]\label{def:nl}
A language $\LL\subseteq\{0,1\}^*$ is in $\NL$ if and only if there exists a deterministic logspace Turing machine $M$ and a polynomial function $q(.)$ such that $$\xx\in\LL\iff \exists \uu\in\{0,1\}^{q(|\xx|)}~\text{ and }~M(\xx,\uu)=1,$$ where $\uu$ is given on a special tape that can be read only from left to right, and $M(\xx,\uu)$ denotes the output of $M$ when $\xx$ is placed on the input tape and $\uu$ is placed on the one-way tape, and $M$ uses at most $\Oh(\log{|\xx|})$ space on its read/write work tape on every input $\xx$. 
\end{definition}

Let $\LL\subseteq\{0,1\}^*$ be a language in $\NL$ and let $M$ be the Turing machine that accepts certificates of $\LL$ as in the previous definition. Then, for each fixed input $\xx$ the set of certificates is accepted by a one-pass logspace Turing machine and therefore their convex hull has extension complexity upper bounded by a polynomial with the degree of the polynomial depending on the constant of the logspace use of the work tape by $M$. Therefore, we have the following.

\begin{theorem}
 Let $\LL\in\NL$ be a language and let $M$ be the Turing machine accepting certificates of $\LL$ as in Definition \ref{def:nl}. For any instance $\xx\in\{0,1\}^n$ the polytope $P_{(\LL,M)}(\xx)=\conv\{\yy\in\{0,1\}^{q(n)}~|~M(\xx,\yy)=1\}$ has polynomial extension complexity. 
\end{theorem}

\subsection{Streaming lower bounds}
Reading Theorem \ref{thm:online_ntm} in contrapositive immediately yields lower bounds in the streaming model of computation. We illustrate this by an example.

\begin{example}
 We know that the perfect matching polytope of the complete graph $K_n$ has extension complexity $2^{\Omega(n)}$ \cite{stoc/Rothvoss14}. Any $p(n)$-pass algorithm requiring space $s(n)$, that correctly determines whether a given stream of $\binom{n}{2}$ is the characteristic vector of a perfect matching in $K_n$, must have $p(n)s(n)=\Omega(n).$ This bound applies even to non-deterministic algorithms.
\end{example}

In fact Lemma \ref{lem:xc_walksignatures} provides an even stronger lower bound.

\begin{definition}
 Let $\LL\subseteq\{0,1\}^*$ be a language. $\LL$ is said to be online $\mu$-magic if there exists a Turing machine $T$ that accepts $\LL$ with the following oracle access. On an input of length $n$ on the one-way input tape, the machine $T$ scans the input only once. At any time (possibly multiple times) during the scanning of the input, $T$ may prepare its working tape to describe\footnote{The description is required only to identify the function uniquely and need not be explicit.} any function $\ff:\{0,1\}^{\mu(n)}\to\{0,1\}^{\mu(n)}$ and a particular input $\xx$ and invoke the oracle that changes the contents of the work-tape to $\ff(\xx)$. The machine must always reject strings not in $\LL$. For strings in $\LL$ there must be some possible execution resulting in accept.
\end{definition}

Notice that the working of even such a machine can be encoded in terms of the configuration graph where the transitions may depend arbitrarily but in a well-formed way on the contents of the work-tape.

\begin{theorem}
 If the set of characteristic vectors of perfect matchings in $K_n$ is accepted by an online $\mu$-magic Turing machine, then $\mu(n)=\Omega(n).$
\end{theorem}

Thus we see that extension complexity lower bounds highlight deep limitations of the streaming model: even powerful oracles do not help solve in sublinear space problems that are $\logspace$ solvable if the one-way restriction on the input is removed.

\subsection{Upper bounds from online algorithms}

\subsubsection*{Parity Polytope}
As an example, consider the language containing strings where the last bit 
indicates the parity of the previous bits. This language can be accepted by a 
deterministic $\logspace$ Turing machine requiring a single pass over the input 
and a single bit of space. Therefore, the parity polytope has extension 
complexity $\Oh(n)$.

The parity polytope is known to have extension complexity at most $4n-4$ \cite{CarrK2005}.

\subsubsection*{Integer Partition Polytope}
For non-negative integer $n$ the Integer Partition Polytope, $\textrm{IPP}_n$, is defined as $\textrm{IPP}_n:=\conv\{x\in\mathbb{Z}^n_{+}|\sum_{k=1}^nkx_k=n\}.$

It is known that $\xc(\textrm{IPP}_n)=\Oh(n^3)$ \cite{dam/OnnS15}. 

% Consider a directed graph with vertices labeled by pair $(i,j)$ with $i,j\in[n]$, and one special node $t$. Create a directed edge from $(i,j)$ to $(i,j+1)$ and label the edge zero. If $i+j+1\leqslant n$, create and edge from $(i,j)$ to $(i+j+1,j+1)$ and label the edge one. If $i+j+1>n$ add an edge from $(i,j)$ to $t$ and label it arbitrarily.
% 
% The resulting network has $n^2+1$ vertices, $2n^2$ arcs and the signatures of length $n$ paths from $(0,0)$ to $(n,n)$ are exactly the vertices of the Integer Partition Polytope $P_n$. This in turn is a face of the Markovian polytope of the constructed network. Therefore, it follows that the Integer Partition polytope has extension complexity $O(n^3).$

Consider the polytope in $\RR^{\lceil\log{n}\rceil\times n}$ that encodes each $x_i$ as a (little-endian) binary string. For example, for $n=4$ the vector $(2,1,0,0)$ is encoded as $(0,1,1,0,0,0,0,0)$. This polytope is clearly an extended formulation of the Integer Partition Polytope. Call this polytope $\textrm{BIPP}_n$. The following single-pass deterministic algorithm accepts a string $(x_1,x_2,\ldots,x_n)\in\{0,1\}^{\lceil\log{n}\rceil\times n}$ if and only if the string represents a vertex of $\textrm{BIPP}_n$.

{%\tiny
\begin{algorithm}[H]
 \KwData{Binary string of length $n\lceil\log{n}\rceil$}
 \KwResult{Accept if the input encodes a vertex of the $\textrm{BIPP}_n$}
 $s=0$; $i=0$; $l=0$\;
 \While{$i<n$}{
  $b=$read\_next\_bit\;
  \eIf{$(s+(i+1)2^lb)>n$}{
   reject\;
   }{
   $s=(s+(i+1)2^lb)$\;
   $l=(l+1)\%\lceil\log{n}\rceil$\;
   \If{$l==0$}{ $i++$\; }{}
  }
 }
 \eIf{ $s==n$ }{ accept\; } { reject\; }

  \caption{One pass algorithm for accepting vertices of $\textrm{BIPP}_n$.}
\end{algorithm}
}

The above algorithm together with Theorem \ref{thm:simulation_lemma_dtm} shows that $\xc(\textrm{IPP}_n)\leqslant \xc(\textrm{BIPP}_n) \leqslant \Oh(n^3\log^2{n}).$ %We should remark that the extended formulation obtained this way contains each $\textrm{IPP}_k$ for $1\leqslant k\leqslant n$ as a face at the expense of a polylogarithmic factor compared to the extended formulation of \cite{dam/OnnS15}.

\subsubsection*{Knapsack Polytopes}
For a given sequence of (non-negative) integers $(a,b)=(a_1,a_2,\ldots,a_n,b)$, the Knapsack polytope $KS(a,b)$ is defined as $KS(a,b):=\left\{x\in\{0,1\}^n\left|\sum_{i=1}^na_ix_i\leqslant b\right.\right\}.$

The Knapsack polytope is known to have extension complexity super-polynomial in $n$. However, optimizing over $KS(a,b)$ can be done via dynamic programming in time $O(nW)$ where $W$ is the largest number among $a_1,\ldots,a_n,b$.

Suppose the integers $a_i,b$ are arriving in a stream with a bit in between indicating whether $x_i=0$ or $x_i=1$. With a space of $W$ bits, an online Turing machine can store and update $\sum_{i=1}^na_ix_i$. At the end, it can subtract $b$ and accept or reject depending on whether the result is $0$ or not. Any overflow during intermediate steps can be used to safely reject the input. Therefore, the extension complexity of the Knapsack polytope is $O(nW\log{W})$. Note however the extension obtained this way is actually an extended formulation of a polytope encoding all the instances together with their solutions.

 \subsubsection*{Languages in co-$\bm{\dlin}$}
 
 $\dlin$ is the class of languages generated by deterministic linear context-free grammars. That is, $\LL\in\dlin$ if and only if there is a deterministic linear context-free grammar generating $\LL$ \cite{colt/HigueraO02}. The following was proved by Babu, Limaye, and Varma \cite{tcs/BabuLRV13}.
 
 \begin{theorem}[BLV]\label{thm:BLV}
 Let $\LL\in\dlin$. Then there exists a probabilistic one-pass streaming algorithm using $\Oh(\log{n})$ space that accepts every string in $\LL$ and rejects every other string with probability at least $1/n^c$ for some constant $c$.
 \end{theorem}
 
 Using the above algorithm together with Theorem \ref{thm:online_ntm} we get the following.
 
 \begin{proposition}
 If $\LL\in\dlin$, then $\overline{\LL}\in\cf.$
 \end{proposition}
\begin{proof}
Let $\LL\in\dlin$. By Theorem \ref{thm:BLV} there exists a probabilistic one-pass algorithm using $\Oh(\log{n})$ space such that if $\xx\in\LL$ then $\xx$ is accepted and if $\xx\notin\LL$ then $\xx$ is rejected with some non-zero probability. Consider this streaming algorithm implemented as a one-pass non-deterministic Turing machine $M$ where the non-deterministic choices of $M$ depend on the probabilistic choices of the algorithm. Now exchange the accept and reject states of $M$ so that $M$ accepts $\xx$ if and only if the algorithm rejects $\xx$. If $\xx\notin\overline{\LL}$ then $\xx\in\LL$ and $M$ rejects $\xx$ for all non-deterministic choices. If $\xx\in\overline{\LL}$ then $\xx\notin\LL$ and there is some non-deterministic choice that leads to $\xx$ being accepted by $M$. Therefore, $\overline{\LL}\in\nspace{1}{\Oh(\log{n})}$ and by Theorem \ref{thm:online_ntm} we have that $\overline{\LL}\in\cf.$
\end{proof}
\section{Conclusion and Outlook}
We have initiated a study of extension complexity of formal languages in this article. We have shown various closure properties of compact languages. This is only a first step in what we hope will be a productive path. We have proved a sufficient machine characterization of compact languages in terms of acceptance by online Turing machines. This property is clearly not necessary. What -- in terms of computational complexity -- characterizes whether or not a language can be represented by small polytopes? We do not know (yet).
%\input{languages-core-cf}
%%%%%%%%%%%%% end languages-core.tex

\section*{Acknowledgements}
The author would like to acknowledge the support of grant GA15-11559S of GA{\v{C}R}. We also thank Mateus De Oliveira Oliveira for finding a critical flaw in a previous proof of Theorem \ref{thm:closure_kleene_star} and the anonymous referees for many valuable suggestions.
%%%%%%%%%%%%% end languages-core.tex

%\section*{References}
\bibliographystyle{elsarticle-num}
\bibliography{languages}

\end{document}